\documentclass[9pt]{article}
\usepackage{spconf,amsmath,graphicx,amssymb,amsthm}
\usepackage{algorithm, algorithmic}

\newcommand{\R}{\mathbb{R}}

\def\minim{\mathop{\hbox{minimize}}}

\newtheorem{thm}{Theorem}[section]

\newtheorem{prop}[thm]{Proposition}

\theoremstyle{remark}
\newtheorem{rem}{Remark}[thm]
\theoremstyle{definition}

\title{SUPPORT DRIVEN REWEIGHTED $\ell_1$ MINIMIZATION}
\name{Hassan Mansour and {\"O}zg{\"u}r Y{\i}lmaz \thanks{hassanm@cs.ubc.ca; oyilmaz@math.ubc.ca}\vspace{-0.2in}
\thanks{
Both authors were supported in part by the Natural
Sciences and Engineering Research Council of Canada (NSERC)
Collaborative Research and Development Grant DNOISE II
(375142-08). {\"O}. Y{\i}lmaz was also supported in part by an NSERC
Discovery Grant.}}
\address{University of British Columbia, Vancouver - BC, Canada}
\begin{document}
%
\maketitle
\begin{abstract}
  In this paper, we propose a support driven reweighted $\ell_1$
  minimization algorithm (SDRL1) that solves a sequence of weighted
  $\ell_1$ problems and relies on the support estimate accuracy. Our
  SDRL1 algorithm is related to the IRL1 algorithm proposed by
  Cand{\`e}s, Wakin, and Boyd. We demonstrate that it is sufficient to
  find support estimates with \emph{good} accuracy and apply constant
  weights instead of using the inverse coefficient magnitudes to
  achieve gains similar to those of IRL1. We then prove that given a
  support estimate with sufficient accuracy, if the signal decays
  according to a specific rate, the solution to the weighted
  $\ell_1$ minimization problem results in a support estimate with
  higher accuracy than the initial estimate. We also show that under
  certain conditions, it is possible to achieve higher estimate
  accuracy when the intersection of support estimates is
  considered. We demonstrate the performance of SDRL1 through numerical
  simulations and compare it with that of IRL1 and standard
  $\ell_1$ minimization.
\end{abstract}
\section{Introduction}
\label{sec:intro}

Compressed sensing is a relatively new paradigm for the acquisition of signals that admit sparse or nearly sparse representations using fewer linear measurements than their ambient dimension \cite{Donoho2006_CS, CRT05}. 

Consider an arbitrary signal $x \in \R^N$ and let $y \in \R^n$ be a
set of measurements given by $y = Ax + e,$ where $A$ is a known
$n\times N$ measurement matrix, and $e$ denotes additive noise that
satisfies $\|e\|_2\leq \epsilon$ for some known $\epsilon\ge
0$. Compressed sensing theory states that it is possible to recover
$x$ from $y$ (given $A$) even when $n \ll N$, that is, using very few
measurements. When $x$ is strictly sparse---i.e., when there are only
$k < n$ nonzero entries in $x$---and when $e=0$, one may recover an
estimate $\hat{x}$ of the signal $x$ by solving the constrained $\ell_0$
minimization problem
\begin{equation}\label{eq:L0_min}
\minim_{u\in \R^N}\ \|u\|_0 \ \text{subject to} \ \ Au=y.
\end{equation}
%
However, $\ell_0$ minimization is a combinatorial
problem and quickly becomes intractable as the dimensions
increase. Instead, the convex relaxation
$$\minim_{u \in \R^N}\ \|u\|_1 \ \text{subject to} \ \|Au - y\|_2 \leq \epsilon \qquad \textrm{(BPDN)}$$ also known as \emph{basis pursuit denoise} (BPDN) can be used to
recover an estimate $\hat{x}$.  Cand{\'e}s, Romberg and Tao \cite{CRT05}
and Donoho \cite{Donoho2006_CS} show that 
(BPDN) can stably and robustly recover $x$ from inaccurate and what
appears to be ``incomplete'' measurements $y = Ax + e$ if $A$ is an
appropriate measurement matrix, e.g., a Gaussian random matrix such
that $n \gtrsim k\log(N/k)$. Contrary to $\ell_0$ minimization, (BPDN)
is a convex program and can be solved efficiently. Consequently, it is
possible to recover a stable and robust approximation of $x$ by
solving (BPDN) instead of \eqref{eq:L0_min} at the cost of increasing
the number of measurements taken.

Several works in the literature have proposed alternate algorithms
that attempt to bridge the gap between $\ell_0$ and $\ell_1$
minimization. These include using $\ell_p$ minimization with $0 < p <
1$ which has been shown to be stable and robust under weaker
conditions than those of $\ell_1$ minimization, see
\cite{gribonval07:_highl, chartrand2008rip, Saab_ellp:2010}.
\emph{Weighted $\ell_1$ minimization} is another alternative
if there is prior information regarding the support of the signal
to-be-receovered as it incorporates such information into the recovery
by weighted basis pursuit denoise (w-BPDN) 
$$\minim_{u}\ \|u\|_{1,\mathrm{w}}\ \text{subject to}\ \|Au - y\|_2 \leq \epsilon, \quad \textrm{(w-BPDN)}$$
where $\mathrm{w}\in [0,1]^N$ and $\|u\|_{1,\mathrm{w}} := \sum_i
\mathrm{w}_i |u_i|$ is the weighted $\ell_1$ norm (see
\cite{CS_using_PI_Borries:2007,
  Vaswani_Lu_Modified-CS:2010,FMSY:2011}). Yet another alternative,
the \emph{iterative reweighted $\ell_1$ minimization} (IRL1) algorithm
proposed by Cand\`{e}s, Wakin, and Boyd \cite{candes2008irl1} and
studied by Needell \cite{Needell:2009} solves a sequence of weighted
$\ell_1$ minimization problems with the weights $\mathrm{w}_i^{(t)}
\approx 1/\left|x_i^{(t-1)}\right|$, where $x_i^{(t-1)}$ is the
solution of the $(t-1)$th iteration and $\mathrm{w}_i^{(0)} = 1$ for
all $i \in \{1\dots N\}$.

In this paper, we propose a support driven iterative reweighted
$\ell_1$ (SDRL1) minimization algorithm that uses a small number of
support estimates that are updated in every iteration and applies a
constant weight on each estimate. The algorithm, presented in section \ref{sec:SDRL1}, relies on the accuracy
of each support estimate as opposed to the coefficient magnitude to
improve the signal recovery. While we still lack a proof that SDRL1
outperforms $\ell_1$ minimization, we present two results in section \ref{sec:theoretical} that
motivate SDRL1 and could lead towards such a proof. First, we prove that
if $x$ belongs to a class of signals that satisfy certain decay
conditions and given a support estimate with accuracy larger than
$50\%$, solving a weighted $\ell_1$ minimization problem with
constant weights is guaranteed to produce a support estimate with
higher accuracy. Second, we show that under strict conditions related to
the distribution of coefficients in a support estimate, it is possible
to achieve higher estimate accuracy when the intersection of support
estimates is considered. Finally, we demonstrate through numerical
experiments in section \ref{sec:Numerical} that the performance of our proposed algorithm is similar
to that of IRL1. \vspace{-0.2in}

\section{Iterative reweighted $\ell_1$ minimization}\label{sec:SDRL1}\vspace{-0.1in}
In this section, we give an overview of the IRL1 algorithm, proposed
by Cand\`{e}s, Wakin, and Boyd \cite{candes2008irl1} and present our
proposed support driven reweighted $\ell_1$ (SDRL1) algorithm. \vspace{-0.2in}

\subsection{The IRL1 algorithm}
IRL1 algorithm solves a sequence of (w-BPDN) problems where the
weights are chosen according to $ \mathrm{w}_i =
\frac{1}{\left|\tilde{x}_i\right| + a}.  $ Here $\tilde{x}_i$ is an
estimate of the signal coefficient at index $i$ (from the previous
iteration) and $a$ is a stability parameter. The choice of $a$ affects
the stability of the algorithm and different variations are proposed
for the sparse, compressible, and noisy recovery cases. The algorithm
is summarized Algorithm \ref{alg:IRL1}.
\begin{algorithm}\caption{IRL1 algorithm \cite{candes2008irl1}}
\begin{algorithmic}[1]\label{alg:IRL1}
\STATE \textbf{Input} $y = Ax + e$
\STATE \textbf{Output} $x^{(t)}$
\STATE \textbf{Initialize} $\mathrm{w}_i^{(0)} = 1$ for all $i \in \{1\dots N\}$, $a$\\
\quad\quad\quad\quad $t = 0$, $x^{(0)} = 0$  
\WHILE{$\|x^{(t)} - x^{(t-1)}\|_2 \leq \textrm{Tol}\|x^{(t-1)}\|_2$}
\STATE $t = t + 1$
\STATE $x^{(t)} = \arg\min\limits_{u} \|u\|_{1,W}$ s.t. $\|Au - y\|_2 \leq \epsilon$
\STATE $\mathrm{w}_i = \frac{1}{\left|\tilde{x}_i\right| + a}$
\ENDWHILE
\end{algorithmic}
\end{algorithm}

The rationale behind choosing the weights inversely proportional to
the estimated coefficient magnitude comes from the fact that large
weights encourage small coefficients and small weights encourage large
coefficients. Therefore, if the true signal were known exactly, then
the weights would be set equal to $\mathrm{w}_i =
\frac{1}{|x_i|}$. Otherwise, weighting according to an approximation
of the true signal and iterating was demonstrated to result in better
recovery capabilities than standard $\ell_1$ minimization. In
\cite{Needell:2009}, the error bounds for IRL1 were shown to be
tighter than those of standard $\ell_1$ minimization. However, aside
from empirical studies, no provable results have yet been obtained to
show that IRL1 outperforms standard $\ell_1$. \vspace{-0.3in}

\subsection{Support driven reweighted $\ell_1$ (SDRL1) algorithm}
In \cite{FMSY:2011}, we showed that solving the weighted $\ell_1$
problem with constant weights applied to a support estimate set
$\widetilde{T}$ has better recovery guarantees than standard $\ell_1$
minimization when the $\widetilde{T}$ is at least 50\%
accurate. Moreover, we showed in \cite{Mansour_YilmazSPIE:2011} that
using multiple weighting sets improves on our previous result when
additional information on the support estimate accuracy is available.
Motivated by these works, we propose the SDRL1 algorithm, a support
accuracy driven iterative reweighted $\ell_1$ minimization algorithm,
which identifies two support estimates that are updated in every
iteration and applies constant weights on these estimates. The
SDRL1 algorithm relies on the support estimate accuracy as opposed to
the coefficient magnitude. The algorithm is presented in Algorithm
\ref{alg:SDRL1}.
\begin{algorithm}\caption{Support driven reweighted $\ell_1$ (SDRL1) algorithm.}
\begin{algorithmic}[1]\label{alg:SDRL1}
\STATE \textbf{Input} $y = Ax + e$
\STATE \textbf{Output} $x^{(t)}$
\STATE \textbf{Initialize} $\hat{p} = 0.99$, $\hat{k} = n\log(N/n)/2$, 

\quad\quad\quad\quad $\omega_1 = 0.5$, $\omega_2 = 0$, Tol, $T_1 = \emptyset$, $\Omega = \emptyset$, 

\quad\quad\quad\quad $t = 0$, $s^{(0)} = 0$, $x^{(0)} = 0$  
\WHILE{$\|x^{(t)} - x^{(t-1)}\|_2 \leq \textrm{Tol}\|x^{(t-1)}\|_2$}
\STATE $t = t + 1$
\STATE $\Omega = \textrm{supp}(x^{(t-1)}|_{s^{(t-1)}}) \cap T_1$ 
\STATE Set the weights equal to $${\small
\mathrm{w}_i = \left\{
\begin{array}{l}
	1, \quad i \in T_1^c \cap \Omega^c \\
	\omega_1, \quad i \in T_1\cap \Omega^c\\
	\omega_2, \quad i \in \Omega
\end{array} \right.}
$$
\STATE $x^{(t)} = \arg\min\limits_{u} \|u\|_{1,\mathbf{w}}$ s.t. $\|Au - y\|_2 \leq \epsilon$
\STATE $l = \min\limits_{\Lambda} |\Lambda|$ s.t. $\|x^{(t)}_{\Lambda}\|_2 \geq \hat{p}\|x^{(t)}\|_2$,\\ $s^{(t)} = \min\{l,\hat{k}\}$
\STATE $T_1 = \textrm{supp}(x^{(t)}|_{s^{(t)}})$
\ENDWHILE
\end{algorithmic}
\end{algorithm}
\noindent

Note that we use two empirical parameters to control the size of the
support estimate $T_1$. The first parameter $\hat{k}$ approximates the
minimum sparsity level recoverable by (BPDN). The second parameter $l$
is the number of largest coefficients of $x^{(t)}$ that contribute an
ad hoc percentage $\hat{p}$ of the signal energy. The size of $T_1$ is
set equal to the minimum of $\hat{k}$ and $l$. \vspace{-0.3in}

\section{Motivating theoretical results}\label{sec:theoretical}
The SDRL1 algorithm relies on two main premises. The first is the
ability to improve signal recovery using a sufficiently accurate
support estimate by solving a weighted $\ell_1$ minimization problem
with constants weights. The second is the intersection set of two
support estimates has at least the higher accuracy of either set.

Let $x \in \R^N$ be an arbitrary signal and suppose we collect $n \ll N$ linear measurements $y = Ax$, $A\in \R^{n \times N}$ where $n$ is small enough (or $k$ is large enough) that it is not possible to recover $x$ exactly by solving (BPDN) with $\epsilon = 0$. Denote by $\hat{x}$ the solution to (BPDN), and by $\hat{x}^\omega$ the solution to (w-BPDN) with weight $\omega$ applied to a support estimate set $\widetilde{T}$. Let $x_k$ be the best $k$-term approximation of $x$ and denote by $T_0 = \textrm{supp}(x_k)$ the support set of $x_k$.
\begin{prop}\label{prop:wL1_guarantee}
  Suppose that $\widetilde{T}$ is of size $k$ with accuracy (with
  respect to $T_0$) $\alpha_0 = \frac{s_0}{k}$ for some integer $k/2 <
  s_0 < k$. If $A$ has RIP with constant $\delta_{(a+1)k} <
  \frac{a-\gamma^2}{a+\gamma^2}$ for some $a > 1$ and $\gamma = \omega
  + (1-\omega)\sqrt{2-2\alpha_0}$, and if there exists a positive
  integer $d_1$ such that
	\begin{equation}\label{eq:wL1_decay}
		|x(s_0 + d_1)| \geq (\omega \eta + 1)\|x_{T_0^c}\|_1 + (1-\omega)\eta\|x_{T_0^c \cap \widetilde{T}^c}\|_1,
	\end{equation} 
	where $\eta = \eta_{\omega}(\alpha_0)$ is a well behaved
        constant, then the set $S = \mathrm{supp}(x_{s_0+d_1})$ is
        contained in ${T}_{\omega} =
        \mathrm{supp}(\hat{x}^{\omega}_{k})$.
\end{prop}

\begin{rem} The constant $\eta_{\omega}(\alpha)$ is given explicitly by
{\small	
		$$
		\eta_{\omega}(\alpha) = \frac{2\left(\sqrt{1 + \delta_{ak}} + \sqrt{a}\sqrt{1 - \delta_{(a+1)k_0}}\right)}{\sqrt{a}\sqrt{1 - \delta_{(a+1)k}} - (\omega + (1-\omega)\sqrt{2 - 2\alpha})\sqrt{1 + \delta_{ak}}}.
		$$
	}
\end{rem}
\begin{proof}[Proof outline]
	The proof of Proposition \ref{prop:wL1_guarantee} is a direct extension of our proof of Proposition 3.2 in \cite{Mansour_YilmazSPIE:2011}. In particular, we want to find the conditions on the signal $x$ and the matrix $A$ which guarantee that the set $S = \mathrm{supp}(x_{s_0+d_1})$ is a subset of $T_{\omega} = \mathrm{supp}(\hat{x}^{\omega}_{k})$. This is achieved when $\hat{x}^{\omega}$ satisfies
\begin{equation}\label{eq:x_omega_condition}
	\min_{j\in S} |\hat{x}^{\omega}(j)| \geq \max_{j\in {T}_{\omega}^c} |\hat{x}^{\omega}(j)| .
\end{equation}
Since $A$ has RIP with $\delta_{(a+1)k} < \frac{a-\gamma^2}{a+\gamma^2}$, it has the Null Space property (NSP) \cite{cohen2006csa} of order $k$, i.e., for any $h \in \mathcal{N}(A)$, $Ah = 0$, then 
$
	\|h\|_1 \leq c_0\|h_{T_0^c}\|_1$, with $c_0 = 1 + \frac{\sqrt{1 + \delta_{ak}}}{\sqrt{a}\sqrt{1 - \delta_{(a+1)k}}}.$\\
Define $h = \hat{x}^{\omega} - x$, then $h \in \mathcal{N}(A)$ and one can show that 
\begin{equation}\label{eq:L1_L1_instopt_wBPDN}
	\|h\|_1 \leq \eta\left(\omega\|x_{T_0^c}\|_1 + (1-\omega)\|x_{T_0^c\cap \widetilde{T}^c}\|_1\right)
\end{equation}
In other words, (w-BPDN) is $\ell_1$-$\ell_1$ instance optimal with these error bounds. The proof of this fact is a direct extension of the $\ell_1$-$\ell_1$ instance optimality of (BPDN) as shown in \cite{cohen2006csa} and we omit the details here. Next, we rewrite \eqref{eq:L1_L1_instopt_wBPDN} as 
$$
	\|h_{T_0}\|_1 \leq (\omega\eta + 1)\|x_{T_0^c}\|_1  + (1-\omega)\eta\|x_{T_0^c\cap \widetilde{T}^c}\|_1 - \|\hat{x}^{\omega}_{T_0^c}\|_1.
$$
To complete the proof, we make the following observations:\\
$(i) \quad \min\limits_{j \in S} |\hat{x}^{\omega}(j)| \geq \min\limits_{j \in S}|x(j)| - \max\limits_{j \in S} |x(j) - \hat{x}^{\omega}(j)|,$\\
$(ii) \quad \|\hat{x}^{\omega}_{T_0^c}\|_1 \geq \max_{j\in {T}_{\omega}^c} |\hat{x}^{\omega}(j)|$\\
which after some manipulations --see \cite{Mansour_YilmazSPIE:2011}, proof of Prop. 3.2 for details of a similar calculation in a different setting-- imply
\begin{equation}\label{eq:RecoveryCondition}
\begin{array}{ll}
	\min\limits_{j \in S} |\hat{x}^{\omega}(j)| \geq & \max_{j\in {T}_{\omega}^c} |\hat{x}^{\omega}(j)|+ \min\limits_{j \in S}|x(j)|  \\ &- (\omega \eta + 1)\|x_{T_0^c}\|_1 + (1-\omega)\eta\|x_{T_0^c \cap \widetilde{T}^c}\|_1.
\end{array}
\end{equation}
Finally, we observe from \eqref{eq:RecoveryCondition} that \eqref{eq:x_omega_condition} holds, i.e., $S \subseteq T_{\omega}$, if 
{\small
$$
	|x(s_0 + d_1)| \geq (\omega \eta + 1)\|x_{T_0^c}\|_1 + (1-\omega)\eta\|x_{T_0^c \cap \widetilde{T}^c}\|_1.
$$}
\end{proof}
Proposition \ref{prop:wL1_guarantee} shows that if the signal $x$ satisfies condition \eqref{eq:wL1_decay} and $\frac{s_0}{k} > 0.5$, then the support of the largest $k$ coefficients of $\hat{x}^{\omega}$ contains at least the support of the largest $s_0 + d_1$ coefficients of $x$ for some positive integer $d_1$.

\begin{figure*}[ht]
	\centering
	\includegraphics[width = 6in, height=1.6in]{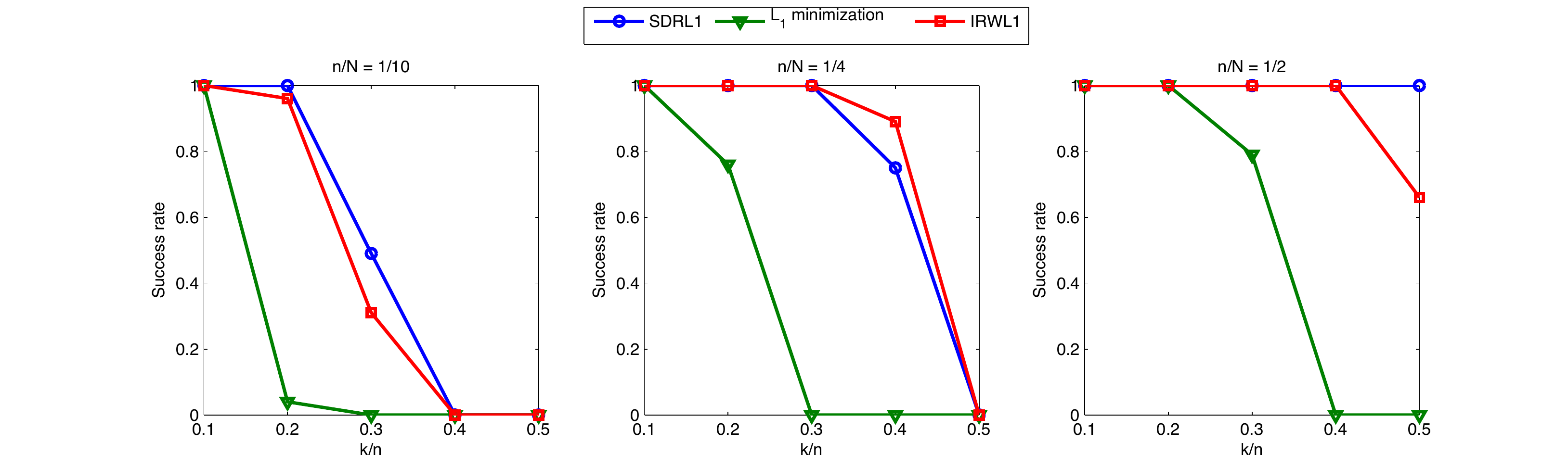}
	\caption{Comparison of the percentage of exact recovery of \emph{sparse signals} between the proposed SDRL1 algorithm, IRL1 \cite{candes2008irl1}, and standard $\ell_1$ minimization. The signals have an ambient dimension $N = 2000$ and the sparsity and number of measurements are varied. The results are averaged over 100 experiments.}\label{fig:sparse_recovery}
\end{figure*}

\begin{figure*}[t]
	\centering
	\includegraphics[width = 6in, height=1.5in]{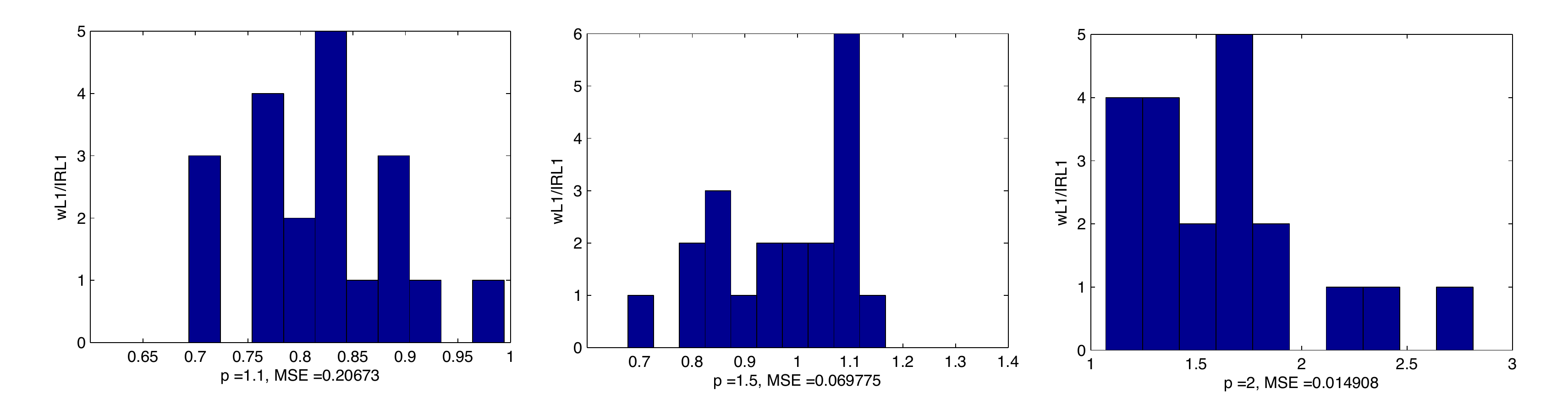}
	\caption{Histogram of the ratio of the mean squared error (MSE) between the proposed SDRL1 and IRL1 \cite{candes2008irl1} for the recovery of \emph{compressible signals}. The signals $x$ follow a power law decay such that $|x_i| = ci^{-p}$, for constant $c$ and exponent $p$.}\label{fig:compressible_histogram}
\end{figure*}

Next we present a proposition where we focus on an idealized scenario:
Suppose that the events $E_i:=\{i\in T\}$, for $i\in \{1,\dots,N\}$ and $T \subseteq \{1,\dots, N\}$, are independent and have equal probability with respect to an appropriate discrete probability measure $\mathrm{P}$. In this case,
we show below, that the accuracy of $\Omega = \widetilde{T}\cap
T_{\omega}$ is at least as high as the higher of the accuracies of
$\widetilde{T}$ and $T_{\omega}$.
For simplicity, we use the notation $\mathrm{P}(T_0 | \widetilde{T})$
to denote $\mathrm{P}(i \in T_0 | i \in \widetilde{T})$.
\begin{prop}\label{prop:support_relations}
	Let $x$ be an arbitrary signal in $\R^N$ and denote by $T_0$ the support of the best $k$-term approximation of $x$. Let the sets $\widetilde{T}$ and $T_{\omega}$ be each of size $k$ and suppose that $\widetilde{T}$ and $T_{\omega}$ contain the support of the largest $s_0$ and $s_1 > s_0$ coefficients of $x$, respectively. Define the set $\Omega = \widetilde{T}\cap T_{\omega}$. Given a discrete probabilty measure $\mathrm{P}$, the events $E_i:=\{i\in T\}$, for $i\in \{1,\dots,N\}$ and $T \subseteq \{1,\dots, N\}$, are independent and equiprobable. Then, for $\rho := \mathrm{P}({T}_{\omega} | \widetilde{T}) \geq \frac{s_0}{k}$, the accuracy of the set $\Omega$ is given by
	
	\centering$\mathrm{P}(T_0|\Omega) = \frac{1}{\rho}\frac{s_0}{k}.$
\end{prop}
\begin{proof}[Proof outline]
	The proof follows directly using elementary tools in probability theory. In particular, we have $\mathrm{P}(T_0|\widetilde{T}) = \frac{s_0}{k}$, and $\mathrm{P}(T_0|T_{\omega}) = \frac{s_1}{k}$. Define $\rho  = \mathrm{P}({T}_{\omega} | \widetilde{T}) \geq \frac{s_0}{k}$, it is easy to see that $\mathrm{P}(T_0 \cap \widetilde{T} | T_0 \cap {T}_{\omega}) = \frac{s_0}{s_1}$ which leads to\\ 
	$\mathrm{P}(T_0 \cap \Omega) = \mathrm{P}(T_0 \cap {T}_{\omega})\mathrm{P}(T_0 \cap \widetilde{T} | T_0 \cap {T}_{\omega}) = \frac{s_1}{N}\frac{s_0}{s_1} = \frac{s_0}{N}$.\\
	Consequently, $\mathrm{P}(T_0|\Omega) = \frac{\mathrm{P}(T_0 \cap \Omega)}{\mathrm{P}({T}_{\omega} | \widetilde{T})\mathrm{Pr}(\widetilde{T})} = \frac{s_0/N}{\rho (k/N)} = \frac{1}{\rho}\frac{s_0}{k}$.
\end{proof}\vspace{-0.2in}
Proposition \ref{prop:support_relations} indicates that as $\mathrm{Pr}({T}_{\omega} |
\widetilde{T}) \rightarrow \frac{s_0}{k}$, then
$\mathrm{Pr}(T_0|\Omega) \rightarrow 1$. Therefore, when $x$ satisfies
\eqref{eq:wL1_decay} it could be beneficial to solve a weighted
$\ell_1$ problem where we can take advantage of the possible
improvement in accuracy on the set $\widetilde{T} \cap
{T}_{\omega}$. Finally, we note that there are more complex
dependencies between the entries of $\widetilde{T}$ and $T_{\omega}$
of Algorithm \ref{alg:SDRL1} for which Proposition
\ref{prop:support_relations} does not account. \vspace{-0.2in}

\section{Numerical results}\label{sec:Numerical}
We tested our SDRL1 algorithm by comparing its performance with IRL1 and standard $\ell_1$ minimization in recovering synthetic signals $x$ of dimension $N = 2000$. We first recover sparse signals from compressed measurements of
$x$ using matrices $A$ with i.i.d. Gaussian random entries and dimensions $n
\times N$ where $n \in \{N/10, N/4, N/2 \}$. The sparsity of the signal is varied such that $k/n \in \{0.1, 0.2, 0.3, 0.4, 0.5\}$. To quantify the reconstruction performance, we plot in Figure \ref{fig:sparse_recovery} the percentage of successful recovery averaged over 100
realizations of the same experimental conditions. The figure shows that both the proposed algorithm and IRL1 have a comparable performance which is far better than standard $\ell_1$ minimization. 

Next, we generate compressible signals with power law decay such that $x(i) = ci^{-p}$ for some constant $c$ and decay power $p$. We consider the case where $n/N = 0.1$ and the decay power $p \in \{1.1, 1.5, 2\}$ and plot the ratio of the reconstruction error of SDRL1 over that of IRL1. Figure \ref{fig:compressible_histogram} shows the histograms of the ratio for 100 experiments each. Note that a ratio smaller than one means that our algorithm has a smaller reconstruction error than that of IRL1. The histograms indicate that both algorithms have a comparable performance for signals with different decay rates.\vspace{-0.2in}

\small
\bibliographystyle{IEEEbib}
\bibliography{sparse}

\end{document}